\documentclass[%
reprint,
nofootinbib,
aps,
pra,
showkeys
]{revtex4-2}
\pdfoutput=1

\usepackage{amsthm, amsmath, amssymb, amsfonts, dsfont}
\usepackage{tensor, braket, physics} 

\usepackage{orcidlink}
\graphicspath{ {./images/} }

\usepackage{algorithm}
\usepackage{algorithmicx}
\usepackage[noend]{algpseudocode}

\usepackage{graphicx, placeins}
\usepackage[caption=false]{subfig}

\usepackage{hyperref, xcolor}
\usepackage[nameinlink,capitalize]{cleveref}
\crefname{algorithm}{Protocol}{Protocols}

\hypersetup{
	colorlinks = true,
	linkbordercolor = {white},
	linkcolor={purple},
	citecolor={purple},
	urlcolor={blue}
}

\newcommand{\Unit}{\mathds{1}}
\newcommand{\dbra}[1]{\langle\bra{#1}}
\newcommand{\dket}[1]{\ket{#1}\rangle}
\newcommand{\dbraket}[2]{\langle\bra{#1}\ket{#2}\rangle}

\newtheorem{theorem}{Theorem}

\crefname{enumi}{Step}{Steps}

\makeatletter
\renewcommand{\ALG@name}{Protocol}
\makeatother

\begin{document}

\title{Multiqubit noise deconvolution and characterization}
\author{Simone Roncallo\,\orcidlink{0000-0003-3506-9027}}
	\email[Simone Roncallo: ]{simone.roncallo01@ateneopv.it}
	\affiliation{Dipartimento di Fisica, Università degli Studi di Pavia, Via Agostino Bassi 6, I-27100, Pavia, Italy}
	\affiliation{INFN Sezione di Pavia, Via Agostino Bassi 6, I-27100, Pavia, Italy}
	
\author{Lorenzo Maccone\,\orcidlink{0000-0002-6729-5312}}
	\email[Lorenzo Maccone: ]{lorenzo.maccone@unipv.it}
	\affiliation{Dipartimento di Fisica, Università degli Studi di Pavia, Via Agostino Bassi 6, I-27100, Pavia, Italy}
	\affiliation{INFN Sezione di Pavia, Via Agostino Bassi 6, I-27100, Pavia, Italy}
	
\author{Chiara Macchiavello\,\orcidlink{0000-0002-2955-8759}}
	\email[Chiara Macchiavello: ]{chiara.macchiavello@unipv.it}
	\affiliation{Dipartimento di Fisica, Università degli Studi di Pavia, Via Agostino Bassi 6, I-27100, Pavia, Italy}
	\affiliation{INFN Sezione di Pavia, Via Agostino Bassi 6, I-27100, Pavia, Italy}

\begin{abstract}
    We present a noise deconvolution technique for obtaining noiseless expectation values of noisy observables at the output of multiqubit quantum channels. For any number of qubits or in the presence of correlations, our protocol applies to any mathematically invertible noise model, even when its inverse map is not physically implementable, i.e. when it is neither completely-positive nor trace-preserving. For a generic observable affected by Pauli noise it provides a quadratic speedup, always producing a rescaling of its Pauli basis components. We show that it is still possible to achieve the deconvolution while experimentally estimating the noise parameters, whenever these are unknown (bypassing resource-heavy techniques such as quantum process tomography). We provide a simulation, with examples for both Pauli and non-Pauli channels.
\end{abstract}
\maketitle

\section{Introduction}
Noise in quantum systems can affect the measurement outcome of any observable, modifying the results of measurement-based protocols or procedures, such as state and process tomography \citep{book:Nielsen,art:QUIT_Tomography,art:Mohseni,art:Mataloni} or quantum simulator experiments \citep{art:Georgescu_Simulation}. This has also important consequences in quantum computation: noise sensitivity remains one of the main drawbacks that prevents quantum computers from outperforming their classical counterparts. Several noise mitigation techniques have been considered in the literature, aimed to reduce errors and potential loss of data in the computation process \citep{art:Temme, art:Endo, art:Endo_Review}. Recently, a noise deconvolution technique was illustrated for observables of single-qubit systems \citep{art:QUIT_NoiseDeconvolution}, by means of a tomographic reconstruction formula that acts like a postprocessing operation on the noisy data, without introducing  modifications to the system.

In this paper we discuss a deconvolution technique that applies to any multiqubit (possibly correlated) noise model \citep{art:QUIT_CorrelatedNoise,art:QUIT_Pauli,%
art:QUIT_QuantumCapacity,art:Daems,art:Huang}, provided its inverse map exists, even if not physically implementable \citep{art:Jiang_InvertibleMaps}. We modify the point of view of \citep{art:QUIT_NoiseDeconvolution}, adopting an operational framework that is more suitable for multiqubit implementations. Unlike other approaches in the literature \citep{art:Temme, art:Endo}, our protocol is passively implemented at the data processing stage: it does not require further experimental (or circuital) configurations, nor active modifications of the original system. For this reason, its range of applicability is not limited to quantum computing: it applies to generic quantum measurements on noisy states, e.g. it can be used to reverse open quantum dynamics \citep{art:Lautenbacher, art:Lloyd_PetzRecovery}.
\begin{figure}[b]
	\centering
	\includegraphics[width = 0.475\textwidth]{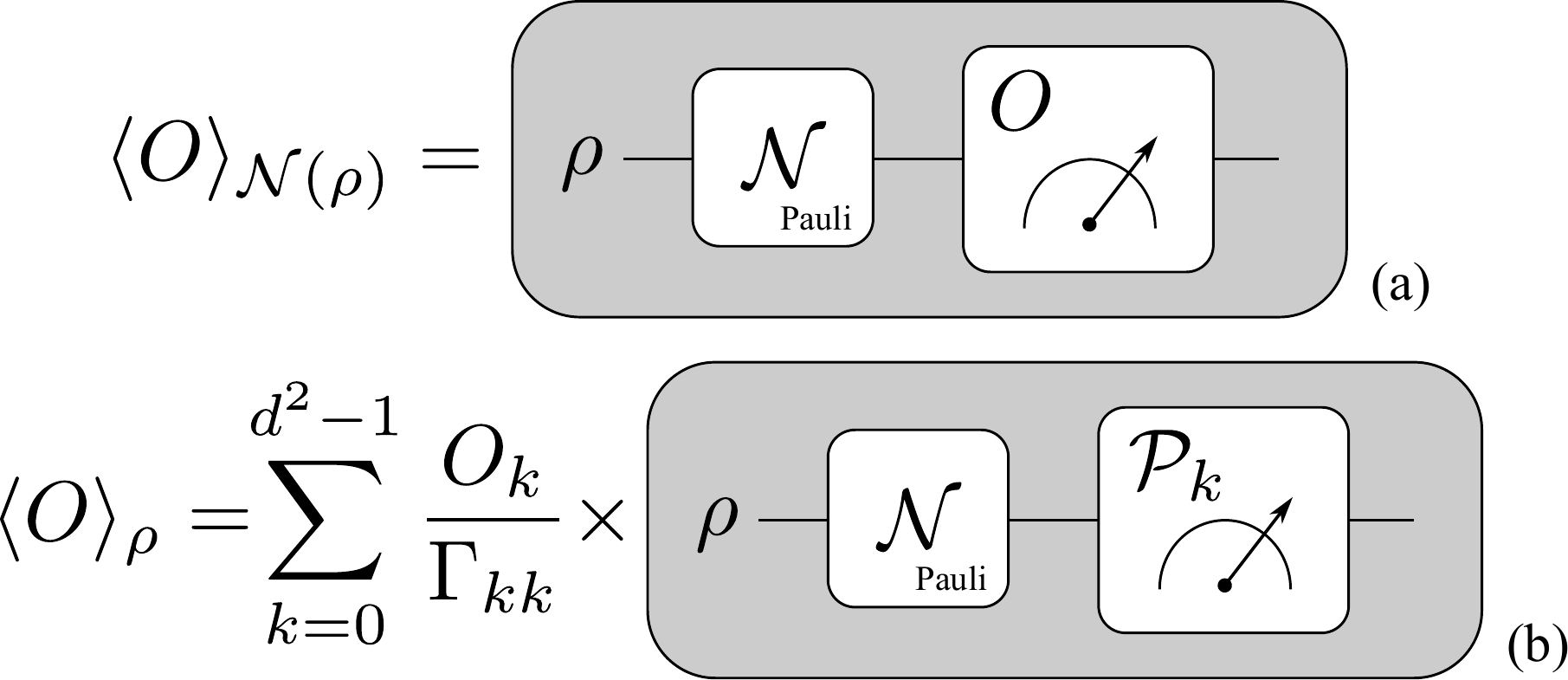}
    \caption{\label{fig:setup}(a) Expectation value of a multiqubit observable $O$, affected by Pauli noise $\mathcal{N}$. (b) The same scheme is applied to the non-zero components of $O$. Deconvolution is achieved by rescaling the noisy data with the corresponding diagonal entry of the Pauli transfer matrix $\Gamma_{\mathcal{N}}$.}
\end{figure}

Our method works in general, e.g. it does not require Markovian correlations. In the specific but important case of Pauli channels \citep{art:Flammia}, it provides a quadratic speedup over other reconstructions. Given an observable $O$, we show that its noiseless expectation value can be tomographically reconstructed by performing local measurements on those elements of the Pauli basis for which $O$ takes non-zero components and by rescaling them in terms of the corresponding entries of the Pauli transfer matrix (PTM) of the channel \citep{art:Greenbaum,art:Nielsen_GateTomography,art:Wood_ChannelRepresentations}. This procedure works for any number of qubits and does not require the complete inversion of the noise map, nor the complete calculation of its PTM. For $n$-qubit Pauli channels the computational complexity of our procedure, i.e. the number of factors required to complete the deconvolution, scales with the number $r = 1,2,..., d^2$ of non-zero components of $O$ in the Pauli basis, with $d = 2^n$ the Hilbert space dimension ($r = 1$ represents an observable made by exactly one basis element, while $r = d^2$ represents the worst cases in which $O$ has all non-zero components). This scenario is summarized in \cref{fig:setup}. Then, we discuss a characterization of the noise map that provides the necessary PTM entries in term of a few measurements on the Pauli basis, without running a full process tomography of the system. This represents an efficient alternative whenever the theoretical computation of the PTM is not doable, e.g. when the noise parameters are unknown. Finally, we show that our protocols apply also to non-Pauli channels, with reduced efficiency. In this case, we provide the deconvolution from the complete inversion of the PTM, and we show that it scales as $d^4$ in the worst case. So, in essence, our procedure experiences a quadratic gain ($d^2$ vs $d^4$) for Pauli channels.

In \cref{sec:SEC1} we review the vectorization of operators and channels. In \cref{sec:SEC2} we present noise deconvolution for multiqubit Pauli channels, addressing their characterization whenever the noise parameters are unknown. Then, we generalize our discussion to the non-Pauli case. In \cref{sec:SEC3} we consider the explicit deconvolution of $n$-qubit bit-flip and depolarizing correlated noises, simulating the latter for $n=3$. As a non-Pauli example, we consider a two-qubit amplitude damping correlated channel, which models qubits losses inside devices.

\section{Vectorization\label{sec:SEC1}}
We consider the Hilbert space of an $n$-qubit system. The basis for the set of operators is
\begin{equation}
	\big\{\sigma_{\alpha_1} \otimes \sigma_{\alpha_2} \otimes ... \otimes \sigma_{\alpha_n} \ | \ \alpha_1, \alpha_2, ..., \alpha_n = 0,1,2,3 \big\} \ ,
	\label{eq:n-basis}
\end{equation}
with $\sigma_0 = \Unit_2$, $\sigma_1 = \sigma_x$, $\sigma_2 = \sigma_y$ and $\sigma_3 = \sigma_z$. We write the Pauli basis in the notation
\begin{equation}
	\big\{ \mathcal{P}_k \ | \ k = 0,1,2,3,...,d^2-1 \big\} \ ,
	\label{eq:pauli-basis}
\end{equation}
with $d = 2^n$ and $\mathcal{P}_k$ denoting the generic element of \cref{eq:n-basis} in lexicographic order. We introduce the vectorized representation \citep{art:Greenbaum}, in which each element of the basis $\mathcal{P}_k$ is mapped to a vector $\ket{k}\rangle$. In this space, an operator $A$ is represented as a $1 \times d^2$ column vector 
\begin{equation}
	\dket{A} = \sum_{k = 0}^{d^2-1} A_k \dket{k} \ ,
\end{equation}
with $A_k = \dbraket{k}{A}$ given by the Hilbert-Schmidt inner product
\begin{equation}
	\dbraket{A}{B} := \frac{1}{d}\Tr[A^{\dagger}B] \ ,
	\label{eq:inner_product}
\end{equation}
with $B$ any operator on this Hilbert space.

Consider a quantum system in the state $\rho$. A quantum channel is a linear completely-positive and trace-preserving (CPTP) map that modifies the system state as $\rho \to \Phi(\rho)$. In the vectorized framework, a channel $\Phi$ is represented by a $d^2 \times d^2$ matrix
\begin{equation}
	\Gamma_{\Phi} = \sum_{j,q = 0}^{d^2-1}\Gamma_{jq}\dket{j}\dbra{q} \ ,
\end{equation}
with components given by the Hilbert-Schmidt inner product
\begin{equation}
	\Gamma_{jq} =  \dbra{j}\Gamma_{\Phi}\dket{q} = \frac{1}{d}\Tr[\mathcal{P}_j \Phi(\mathcal{P}_q)] \ .
	\label{eq:PTMDefinitionVect}
\end{equation}
This is called the Pauli transfer matrix of the channel \citep{art:Greenbaum}. In this representation, the action of the channel $\Phi(A)$ is given by a matrix-vector multiplication
\begin{equation}
	\dket{\Phi(A)} = \Gamma_{\Phi}\dket{A} = \sum_{j,q = 0}^{d^2-1} \Gamma_{jq} A_{q} \dket{j} \ .
\end{equation}
The CPTP condition guarantees that $\Gamma_{0q} = \delta_{0q}$, with $\delta_{jq}$ the Kronecker delta. For unital channels,  i.e. when $\Phi(\Unit) = \Unit$, it holds also that $\Gamma_{j0} = \delta_{j0}$.

Consider a channel $\Phi$, its adjoint $\Phi^*$ is the map satisfying
\begin{equation}
	\dbraket{A}{\Phi^*(B)} = \dbraket{\Phi(A)}{B} \ .
\end{equation}
This implies that the adjoint PTM, here denoted by $\Gamma_\Phi^*$, is precisely the Hermitian conjugate of $\Gamma_\Phi$.

\section{Noise deconvolution\label{sec:SEC2}}
Noise in open quantum systems is modeled in terms of quantum channels \citep{book:Nielsen}, namely, linear CPTP operations $\mathcal{N}$ that map the ideal, i.e. noiseless, state $\rho$ into a noisy one $\rho' = \mathcal{N}(\rho)$. Different choices of $\mathcal{N}$ correspond to different noise models, e.g. the bit-flip, the dephasing, the depolarizing, or the amplitude damping noises \citep{book:Nielsen}. With the state modified by $\mathcal{N}$, any measurement performed on an observable $O$ becomes noisy, namely its expectation value $\langle O \rangle_{\rho}$ is mapped to $\langle O \rangle_{\rho'}$. In this section we introduce noise deconvolution as a technique that provides the ideal expectation value $\langle O \rangle_{\rho}$ of arbitrary operators, using the noisy data obtained on $\rho'$. Our protocol applies to any mathematically invertible noise model, i.e. one for which $\Gamma_{\mathcal{N}}^{-1}$ exists, even when the inverse channel is not physically implementable, i.e. when $\mathcal{N}^{-1}$ is not CPTP.

We consider the noise deconvolution equation derived in \citep{art:QUIT_NoiseDeconvolution}
\begin{equation}
    \left\langle O \right\rangle_{\rho} = \left\langle \mathcal{N}^{*^{-1}}(O) \right\rangle_{\rho'} \ ,
    \label{eq:expval}
\end{equation}
which yields the ideal, i.e. noiseless, expectation value of an arbitrary observable $\left\langle O \right\rangle_{\rho}$ (or even of a non-observable operator) by evaluating, instead, the inverted adjoint map $\mathcal{N}^{*^{-1}}(O)$ over the noisy state $\rho' = \mathcal{N}(\rho)$. In other proposals \citep{art:Beny}, similar inversions are implemented physically (often only approximately) by introducing suitable modifications to the channel. Here, instead, we just use the noisy measured data to calculate the noiseless value. In other words and in contrast to previous proposals, our noise-inversion reconstruction is implemented entirely and solely at the data processing stage.

In the vectorized framework, \cref{eq:expval} reads
\begin{equation}
	\dbraket{\rho}{O} = \langle\bra{\rho '} \Gamma_{\mathcal{N}}^{*^{-1}} \ket{O}\rangle \ ,
	\label{eq:PTM-deconvolution}
\end{equation}
where $\Gamma_{\mathcal{N}}^{*^{-1}}$ is the inverse adjoint PTM. 

We start with the deconvolution of multiqubit Pauli channels, for which the vectorization guarantees a quadratic speedup in efficiency over the general case (treated below). The current method generalizes and supersedes the single-qubit analysis presented in \citep{art:QUIT_NoiseDeconvolution}. In this case we derive a reconstruction formula without completely inverting the noise channel, instead using only some components of the PTM. We first apply this procedure to an observable made by one of the possible $n$-fold tensor products of the Pauli matrices, i.e. one of the elements of the basis, then we extend our considerations to the expectation value of a generic observable that takes all the contributions from the Pauli basis. At the end of the section, we generalize our discussion to the non-Pauli case.

In the Pauli basis, the Kraus representation \citep{book:Nielsen} of an $n$-qubit Pauli channel is \citep{art:Flammia}
\begin{equation}
	\mathcal{N}(\rho) = \sum_{j=0}^{d^2-1}\beta_j \mathcal{P}_j \rho \mathcal{P}_j \ ,
	\label{eq:Pauli-channel_PauliBasis}
\end{equation}
with $\sum_{j}\beta_j = 1$ and $\beta_j \geq 0 \ \forall j$. Such channels represent an example of random unitary maps \citep{art:Audenaert_RandomUnitary}, where each unitary $\mathcal{P}_j$ is applied with probability $\beta_j$. In the case of \cref{eq:Pauli-channel_PauliBasis}, this produces a depolarizing contraction of the Bloch hypersphere, whose intensity and direction depend on the choice of $\beta_j$ \citep{art:Siudzinska_GeneralizedPauliChannels}.

A Pauli channel yields a diagonal PTM
\begin{equation}
	\Gamma_{\mathcal{N}} = \sum_{j=0}^{d^2-1} \lambda_j \dket{j}\dbra{j} \ ,
	\label{eq:DiagonalPTM}
\end{equation}
where $\lambda_j = \dbra{j} \Gamma_{\mathcal{N}} \dket{j}$. As the only requirement, we ask $\Gamma_{\mathcal{N}}$ to be mathematically invertible, i.e. that $\lambda_j \neq 0 \ \forall j$.\footnote{An example of a non-invertible map is represented by a depolarizing channel with $p=1$, for which the Bloch hypersphere collapses to a single point.} For now, we also assume that the noise parameters are known so that a theoretical computation of the PTM is always possible. This last assumption simplifies our analysis, but it is not necessary. We discuss the case of unknown noise in the following paragraphs.

Since all the Kraus operators in \cref{eq:Pauli-channel_PauliBasis} are Hermitian, the adjoint of the Pauli channel is the channel itself \citep{art:QUIT_NoiseDeconvolution}, yielding $\Gamma_{\mathcal{N}}^{*^{-1}} = \Gamma_{\mathcal{N}}^{-1}$.  Then the inverse PTM is diagonal too and its components read $1/\lambda_j$.

First, consider an observable given by the $k$th element of the basis, which in vectorized notation corresponds to a column vector with only one non-zero component $\ket{O}\rangle = \ket{k}\rangle$. The action of the inverse PTM yields
\begin{equation}
	\Gamma^{-1}_{\mathcal{N}} \dket{O}  = \Gamma_{kk}^{-1}\dket{k} = \frac{1}{\lambda_k} \dket{k} \ . 
	\label{eq:InverseDiagonalPTM}
\end{equation}
Using \cref{eq:PTM-deconvolution}, the expectation value follows as
$\dbraket{\rho}{k} = \lambda_k^{-1} \dbraket{\rho'}{k}$, which in standard non-vectorized notation reads
\begin{equation}
	\langle \mathcal{P}_{k} \rangle_{\rho} = \frac{d}{\Tr[\mathcal{P}_k \mathcal{N}(\mathcal{P}_k)]} \langle \mathcal{P}_{k} \rangle_{\rho'} \ .
	\label{eq:deconv}
\end{equation}
This shows that for an $n$-qubit Pauli channel, the deconvolution of the expectation value of the $k$th element of the basis is always obtained as a rescaling of the noisy outcome, which depends on the $k$th element on the diagonal of the inverse PTM.

We now discuss the deconvolution of a generic observable $O$ subject to an arbitrary $n$-qubit Pauli channel. We expand $O$ in terms of the basis vectors $\dket{O} = \sum_k O_k \dket{k}$ to compute the right-hand side of \cref{eq:PTM-deconvolution}. Applying the same strategy, each component $O_k$ must be rescaled by the corresponding element on the diagonal of the inverse PTM. Then, the reconstructed measurement outcome reads
\begin{equation}
	\langle O \rangle_\rho = \sum_{k=0}^{d^2-1} \frac{d}{\Tr[\mathcal{P}_k \mathcal{N}(\mathcal{P}_k)]} O_k \langle \mathcal{P}_k \rangle_{\rho'} \ .
	\label{eq:FinalDeconvolution}
\end{equation}
Namely, we obtain the ideal noiseless outcome (over $\rho$) by processing the noisy expectation values (over $\rho' = \mathcal{N}(\rho)$) of those Pauli basis elements $\mathcal{P}_{k}$ that contribute to the expansion of $O$.

This procedure works for any number $n$ of (even correlated) qubits, and it involves only local measurements on each element of the basis. The computational complexity, i.e. the number of factors required, scales with the number $r=1,2,...,d^2$ of non-zero elements of the expansion of $O$ on the Pauli basis. In the trivial case in which $O$ is exactly an $n$-qubit Pauli matrix, the deconvolution always requires a single measurement and one PTM entry, for any number of qubits. On the other hand, when $O$ is a generic observable, the deconvolution requires a measurement on the entire basis and the computation of all the $d^2$ diagonal components of the PTM. In any case, this considerably reduces the number of computations of the PTM entries over the inversion-based method of \citep{art:QUIT_NoiseDeconvolution}, which for $n$ qubits always requires $d^4$ operations.

So far we have considered $n$-qubit Pauli channels whose parameters are known \emph{a priori}. Although it simplifies the analysis, this assumption is not necessary: we can still estimate these parameters without running a full process tomography of the channel \citep{book:Nielsen,art:Mohseni}, which is intractable for large $n$. For a general $n$-qubit Pauli channel $\mathcal{N}$, with unknown Kraus representation coefficients in \cref{eq:Pauli-channel_PauliBasis}, the deconvolution is achieved by means of the following characterization. Prepare the system in the state 
\begin{equation}
	\rho_k = \frac{\Unit + \mathcal{P}_k}{d} \quad \text{for } k \neq 0 \ ,
	\label{eq:ChInputState}
\end{equation}
with $\Unit$ the $n$-qubit identity operator (see \cref{sec:App} for a discussion on positivity). Assuming that $\rho_k$ evolves to $\rho_k' = \mathcal{N}(\rho_k)$, i.e. that the channel acts independently of the preparation scheme of \cref{eq:ChInputState}, unitality guarantees that
\begin{equation}
	\rho_k' = \frac{\Unit + \mathcal{N}(\mathcal{P}_k)}{d} \ .
	\label{eq:ChNoisyState}
\end{equation}
The diagonal entries of the PTM then read
\begin{align}
	&\Gamma_{kk} =  \langle \mathcal{P}_k \rangle_{\rho_k'}  \quad \text{for } k \neq 0 \ ,
	\label{eq:PTMfromCharachterization}
\end{align}
with $\langle\mathcal{P}_k \rangle_{\rho_k'} = \Tr[\mathcal{P}_k \mathcal{N}(\rho_k)]$, which can be used in the deconvolution formula in \cref{eq:deconv} or \cref{eq:FinalDeconvolution}. This means that, even without knowing the noise parameters, we can still obtain the reconstruction factors by measuring the $k$th element of the Pauli basis over a noisy state initially prepared as in \cref{eq:ChInputState}. Again, this requires at most $r$ operations, with $r = 1,2,...,d^2$ the number of non-zero components of $O$ in the Pauli basis. We summarize the entire procedure in \cref{alg:Deconvolution}.
\begin{algorithm}[H]
\caption{Multiqubit Pauli channel deconvolution\label{alg:Deconvolution}}\vspace{2pt}
\textbf{Input:} Observable $O$ \textbf{Optional:} Pauli channel $\mathcal{N}$\\[2pt]
\textbf{Result:} Noiseless expectation value $\langle O \rangle_{\rho}$\vspace{2pt}
\begin{algorithmic}[1]
\State construction of the basis \Comment{$\{\mathcal{P}_k\}$ for $k = 0,1,...,d^2-1$} \vspace{2pt}%

\State projection of $O$ on the basis \Comment{$O_k = \Tr[O \mathcal{P}_k]/d$} \vspace{2pt}%

\For{$k$ such that $O_k \neq 0$} \vspace{1.5pt}%

	\State{noisy measurement on the basis} \Comment{$\langle \mathcal{P}_k \rangle_{\mathcal{N}(\rho)}$} \vspace{2pt}%
	\If{$\mathcal{N}$ is known} \vspace{2pt}%
	
		\State get PTM \Comment{$\Gamma_{kk} \gets \Tr[\mathcal{P}_k \mathcal{N}(\mathcal{P}_k)]/d$} \vspace{2pt}
	\Else \vspace{2pt}%
		\State state preparation \Comment{$\rho_k = (\Unit + \mathcal{P}_k)/d$} \vspace{2pt}
		\State PTM characterization \Comment{$\Gamma_{kk} \gets \langle \mathcal{P}_k \rangle_{\mathcal{N}(\rho_k)}$}
	\EndIf%
\EndFor
	\State{noise deconvolution} \Comment{$\langle O \rangle_{\rho} \gets \sum_k \Gamma_{kk}^{-1}O_k\langle \mathcal{P}_k \rangle_{\mathcal{N}(\rho)}$}
\end{algorithmic}
\end{algorithm}
Before considering the applications, we discuss the case of noise channels that cannot be expressed in the form of \cref{eq:Pauli-channel_PauliBasis}, i.e. whose Kraus representation contains at least one operator that is different from a Pauli $\mathcal{P}_k$. These channels describe purely quantum mechanical processes, e.g. a ``spontaneous emission'' which is modeled by the amplitude damping channel \citep{book:Nielsen}.

Non-Pauli channels are not self-adjoint and have a non-diagonal PTM, so $\Gamma_{\mathcal{N}}^{*^{-1}}$ cannot be simply obtained as in \cref{eq:InverseDiagonalPTM}. In this case, the tomographic reconstruction formula of a generic observable $O$ reads
\begin{equation}
	\langle O \rangle_\rho = \sum_{j,q=0}^{d^2-1} \big(\Gamma^{*^{-1}}\big)_{jq} O_q \langle \mathcal{P}_j \rangle_{\rho'} \ ,
	\label{eq:NonPauliDeconvolution}
\end{equation}
which, in the case of $O = \mathcal{P}_k$ ($O = \dket{k}$ in the vectorized notation), reduces to
\begin{equation}
	\langle \mathcal{P}_k \rangle_\rho = \sum_{j=0}^{d^2-1} \big(\Gamma^{*^{-1}}\big)_{jk} \langle \mathcal{P}_j \rangle_{\rho'} \ .
	\label{eq:SimpleNonPauliDeconvolution}
\end{equation}
In this case, the deconvolution procedure still works, although less efficiently: it requires more computations and the complete calculation of the inverse adjoint PTM $\Gamma_{\mathcal{N}}^{*^{-1}}$. For Pauli channels, both equations consistently reduce to \cref{eq:deconv} and \cref{eq:FinalDeconvolution}. In the next section we consider the explicit deconvolution of a two-qubit amplitude damping correlated channel.

When noise is not described by a Pauli channel, as it often occurs in real devices, the quadratic speedup of \cref{eq:FinalDeconvolution} is not achievable. However, generic noise models are usually split in terms of Pauli and non-Pauli (but simpler) contributions. See for example \citep{art:QUIT_NoiseDeconvolution}, where the decoherence of Rigetti Aspen-9 is modeled as a composition of a dephasing (Pauli) and an amplitude damping (non-Pauli) channel. In the vectorized representation, the composition of multiple channels reduces to the product of their PTM. Even in this case, their deconvolution can be separately treated, with still a quadratic advantage on the Pauli terms, which then act as a rescaling of the non-Pauli rows or columns. Alternatively, twirling techniques can be employed to remove the off-diagonal elements of the PTM \citep{art:Flammia,art:Cai_PauliConjugation}, thus mapping the original channel to an effective Pauli one.

Similarly to the diagonal case, a straightforward characterization can be employed whenever the channel is unital and its parameters are unknown. Prepare the system in the state $\rho_k$ of \cref{eq:ChInputState} evolved to \cref{eq:ChNoisyState} and then iterate \cref{eq:PTMfromCharachterization} for all the elements $\mathcal{P}_j$ of the Pauli basis, yielding the PTM as
\begin{align}
	&\Gamma_{jk} =  \langle \mathcal{P}_j \rangle_{\rho_k'}  \quad \text{for } j, k \neq 0 \ ,
	\label{eq:NonPauliPTMfromCharachterization}
\end{align}
which can be inverted and substituted in \cref{eq:NonPauliDeconvolution} or \cref{eq:SimpleNonPauliDeconvolution}.

Both \cref{eq:PTMfromCharachterization} and  \cref{eq:NonPauliPTMfromCharachterization} provide a direct tomographic reconstruction of the channel PTM, as long as the channel is unital. This represents an alternative to the standard approach that recovers the channel Kraus representation through standard quantum process tomography (a generalization of this procedure to non-unital channels and a comparison with quantum process tomography will be discussed in \citep{art:QUIT_InPreparation}).
\begin{figure*}[t]
	\centering
	\subfloat[\label{fig:depo3_comp}]{\includegraphics[width = 0.5 \textwidth]{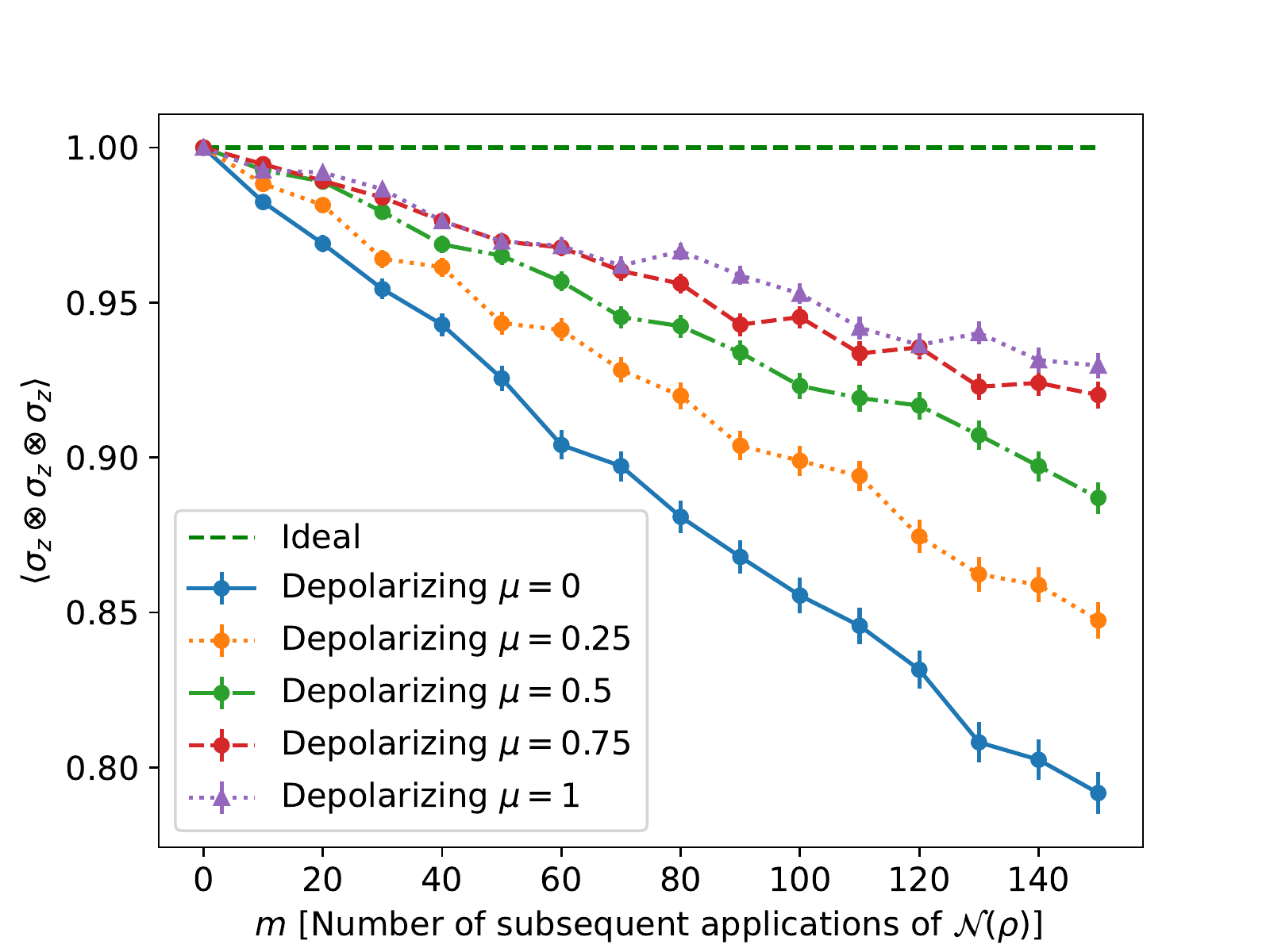}}%
	\subfloat[\label{fig:depo3}]{\includegraphics[width = 0.5 \textwidth]{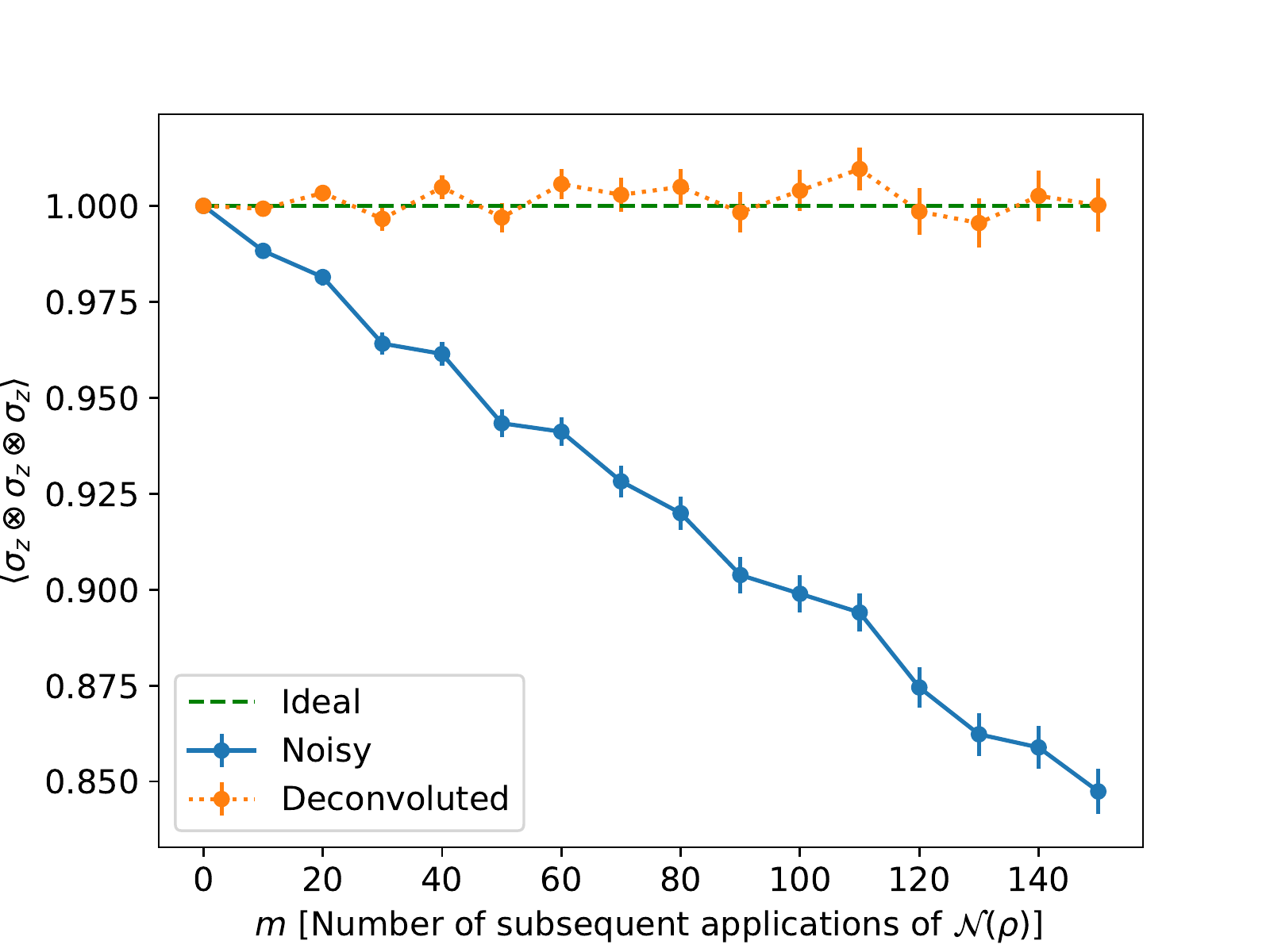}}
	\caption{(a) Simulated noise effect on $\langle \sigma_z^{\otimes3} \rangle$ for a three-qubit depolarizing correlated channel, where the degree of correlation is parametrized by $\mu \in [0,1]$ and with probability $q = 0.00052$. The system is initially prepared in the noiseless state $\rho = \ket{000}\!\bra{000}$. The ideal expectation value $\langle \sigma_z^{\otimes3}\rangle_{\rho} = 1$ is represented by a dashed horizontal line. The simulation is performed using Qiskit Aer for $8192$ shots, with respect to the number $m$ of subsequent applications of the noise map $\mathcal{N}^m(\rho)$. The noise effect is plotted for different values of $\mu$. Note that when the value of $\mu$ decreases, the correlation of qubits increases and the effect of the noise decreases. (b) Simulated deconvolution of the noise channel for $\mu = 0.25$. We compare the noisy output with the noiseless one, recovered with the noise deconvolution procedure presented here.} 
\end{figure*}

\section{Applications\label{sec:SEC3}}
In this section we discuss noise deconvolution of several examples of Pauli and non-Pauli noise models. To analyze correlations we demonstrate our general method on a specific class of channels, in which a parameter $\mu$ measures the amount of correlations \citep{art:QUIT_CorrelatedNoise,art:QUIT_Pauli,art:QUIT_QuantumCapacity,art:QUIT_FullyCorrelatedDamping,art:QUIT_CorrelatedDamping}.   

We start with the class of Pauli correlated channels defined in \citep{art:QUIT_CorrelatedNoise,art:QUIT_Pauli,art:QUIT_QuantumCapacity}, whose Kraus representation reads
\begin{equation}
	\mathcal{N}(\rho) = \sum_{\{\alpha_i\}=0}^{3} p_{\alpha_1 \alpha_2 ... \alpha_n} A_{\alpha_1 \alpha_2 ... \alpha_n} \rho A^{\dagger}_{\alpha_1 \alpha_2 ... \alpha_n} \ ,
	\label{eq:Pauli-channel}
\end{equation}
with $i \in \{1,2,...,n\}$ and Kraus operators
\begin{equation}
	A_{\alpha_1 \alpha_2 ... \alpha_n} = \sigma_{\alpha_1} \otimes \sigma_{\alpha_2} \otimes ... \otimes \sigma_{\alpha_n} \ ,
\end{equation}
with $p_{\alpha_1 \alpha_2 ... \alpha_n} = p_{\alpha_1}p_{\alpha_2|\alpha_1}...p_{\alpha_n|\alpha_{n-1}}$ given by the Markov chain \citep{art:QUIT_QuantumCapacity,art:Hamada}
\begin{align}	
	&p_{\alpha_{j}|\alpha_{i}} = (1-\mu)p_{\alpha_j} + \mu\delta_{\alpha_i \alpha_j} \ , \\
	&\vec{p} = [1-p,p_x,p_y,p_z]^{T} \ ,
\end{align}
with $p = p_x + p_y + p_z$. This kind of correlation is not a requirement of our deconvolution technique, it provides only an example to test our protocol. The parameter $\mu \in [0,1]$ represents the degree of correlation between couple of qubits. On the one hand, $\mu = 0$ represents a memoryless channel, which is when the qubits are completely uncorrelated. On the other hand, $\mu = 1$ describes a full-memory channel, i.e. when the qubits are completely correlated. The Pauli channels generalize noise models such as the bit-flip, the bit-phase-flip, the dephasing and the depolarizing which can be re-obtained from \cref{eq:Pauli-channel} by a proper choice of $\vec{p}$. For example, we obtain  the bit-flip channel when $\vec{p} = [1-p,p,0,0]^T$ or the depolarizing channel when $\vec{p} = [1-p,p/3,p/3,p/3]^T$. 

As the first example, consider the $n$-qubit observable $\sigma_z^{\otimes n}$, whose expectation value is affected by a bit-flip correlated noise \citep{book:Nielsen, art:QUIT_QuantumCapacity}, i.e. with $\vec{p} = [1-p,p,0,0]^T$. In this case, the deconvolution formula yields
\begin{equation} 
\langle \sigma_z^{\otimes n} \rangle_{\rho} = f_n(p,\mu) \langle \sigma_z^{\otimes n} \rangle_{\rho'} \ ,
\end{equation}
where the reconstruction factor is given by
\begin{equation}
	f_n(p,\mu) = \frac{2^n}{\Tr[\sigma_z^{\otimes n}\mathcal{N}(\sigma_z^{\otimes n})]} \ .
\end{equation}%
By direct computation, we obtain
\begin{align}
	& f_{1}(p,\mu) = \frac{1}{1-2p} \ , \\
	& f_{2}(p,\mu) = \frac{1}{1 + 4 (\mu - 1) (1 - p) p} \ , \\
	& f_{3}(p,\mu) = \frac{1}{(1 - 2 p) [1 + 4 (\mu - 1)^2 (p - 1) p]} \ . 
\end{align}
For a depolarizing correlated channel \citep{art:QUIT_CorrelatedNoise,art:QUIT_QuantumCapacity}, i.e. when $\vec{p} = [1-3q/4,q/4,q/4,q/4]^{T}$, the same computation gives
\begin{align}
	& f_{1}(q,\mu) = \frac{1}{1-q} \ , \\
	& f_{2}(q,\mu) = \frac{1}{1 + (\mu -1) (2 - q) q} \ , \\
	& f_{3}(q,\mu) = \frac{1}{(1-q) [1 + (\mu - 1)^2 (q - 2) q]} \label{eq:example} \ . 
\end{align}
When the qubits are completely uncorrelated ($\mu = 0$), the noise deconvolution factorizes in terms of the single-qubit contributions. For qubits that are completely correlated ($\mu = 1$), the noise has no effect when $n$ is even, while it is corrected by a single-qubit contribution when $n$ is odd.
In \cref{fig:depo3_comp} we plot the effect of a three-qubit depolarizing correlated channel on $\langle \sigma_z^{\otimes 3} \rangle$ with respect to the number of subsequent applications of the noise map, parametrized in terms of $\mu$. In \cref{fig:depo3} we perform a simulation of the deconvolution, which successfully reproduces the noiseless expectation value.

As a non-Pauli example, we consider a two-qubit amplitude damping correlated channel \citep{art:QUIT_FullyCorrelatedDamping,art:QUIT_CorrelatedDamping}. For a single-qubit system, the amplitude damping channel is defined in terms of two Kraus operators
\begin{equation}
	E_0 = \begin{pmatrix}
	1 & 0 \\ 0 & \sqrt{\eta}
	\end{pmatrix} \ E_1 = \begin{pmatrix}
	0 & \sqrt{1-\eta} \\ 0 & 0
	\end{pmatrix} ,
\end{equation}
where $1-\eta$ represents the probability of the system losing a qubit, e.g. by emitting a photon \citep{book:Nielsen}, and it plays the role of channel transmissivity, e.g. in the case of optical fibers \citep{art:QUIT_CorrelatedDamping}.

The two-qubit amplitude damping correlated channel can be obtained as a convex combination of a memoryless amplitude damping channel $\mathcal{N}_0$ with a memoryful one $\mathcal{N}_1$ 
\begin{equation}
	\mathcal{N}(\rho) = (1-\mu)\mathcal{N}_0(\rho) + \mu \mathcal{N}_1(\rho) \ ,
\end{equation}
where
\begin{equation}
	\mathcal{N}_0(\rho) = \sum_{j=0}^3 A_j \rho A_j^{\dagger} \ , \quad \mathcal{N}_1(\rho) = \sum_{j=0}^1 B_j \rho B_j^{\dagger} \ , 
\end{equation}
with Kraus operators $A_0 = E_0 \otimes E_0$, $A_1 = E_0 \otimes E_1$, $A_2 = E_1 \otimes E_0$, $A_3 = E_1 \otimes E_1$ and
\begin{equation}
	B_0 = \begin{pmatrix}
	1 & 0 & 0 & 0 \\ 0 & 1 & 0 & 0 \\ 0 & 0 & 1 & 0 \\ 0 & 0 & 0 & \sqrt{\eta}
	\end{pmatrix} \
	B_1 = \begin{pmatrix}
	0 & 0 & 0 & \sqrt{1 - \eta} \\ 0 & 0 & 0 & 0 \\ 0 & 0 & 0 & 0 \\ 0 & 0 & 0 & 0
	\end{pmatrix} .
\end{equation}
Consider the two-qubit observables $\sigma_x^{\otimes 2}$, $\sigma_y^{\otimes 2}$ and $\sigma_z^{\otimes 2}$ with noisy state $\rho' = \mathcal{N}(\rho)$. From the deconvolution formula and the computation of the inverse adjoint PTM, it follows that
\begin{multline}
	\langle \sigma_x^{\otimes 2} \rangle_{\rho} = f(\eta,\mu)  \big\{ \big[ 2\eta(1-\mu) + \mu(\sqrt{\eta}+1) \big]\langle \sigma_x^{\otimes 2} \rangle_{\rho'}  \\
	+ \mu (\sqrt{\eta}-1)\langle \sigma_y^{\otimes 2} \rangle_{\rho'} \big\} \ ,
	\label{eq:AmpDampSigmaXX}
\end{multline} 
while $\langle \sigma_y^{\otimes 2} \rangle_{\rho}$ is obtained from \cref{eq:AmpDampSigmaXX} through the index exchange $x \leftrightarrow y$, and that
\begin{multline}
	\langle \sigma_z^{\otimes 2} \rangle_{\rho} = g(\eta,\mu) \big[ (\mu-1)^2(\eta-1)^2 + \langle \sigma_z^{\otimes 2} \rangle_{\rho'} \\ 
	 -(\mu-1)(\eta-1) \langle \Unit \otimes \sigma_z + \sigma_z \otimes \Unit \rangle_{\rho'} \big]  \ ,
\end{multline}
with 
\begin{align}
	&f(\eta,\mu) = \frac{1}{2\left[\mu\left(\eta-\sqrt{\eta}\right) - \eta\right] \left[\mu(\eta - 1) - \eta\right]} \ , \\
	&g(\eta,\mu) = \frac{1}{\left[ \eta + \mu(1 - \eta) \right]^2} \ .
\end{align}
Note that when the qubits are completely uncorrelated ($\mu = 0$) the noise deconvolution factorizes in terms of the single-qubit contributions.\footnote{See \citep{art:QUIT_NoiseDeconvolution} for the deconvolution of the single-qubit amplitude damping channel, with $\eta = 1 - \gamma$.} When the two qubits are completely correlated ($\mu = 1$) the noise has no effect on $\langle \sigma_z^{\otimes 2} \rangle$.

\section{Conclusions}
We illustrated a procedure for the deconvolution of multiqubit noise described by mathematically invertible quantum channels: it returns the ideal expectation value of arbitrary observables from the noisy data. 

In the case of Pauli channels, our prescription bypasses the inversion of the noise map, providing the deconvolution from a set of Pauli measurements rescaled by a few components of the PTM (where the number of factors is quadratically reduced with respect to the general scenario). As discussed and shown in the simulation, this analysis can be applied to any example of multiqubit Pauli noise, e.g. the bit-flip or the depolarizing correlated channels. Then, we presented a characterization technique that provides the necessary PTM entries as a set of Pauli measurements on a specific class of input states, guaranteeing the deconvolution whenever a theoretical approach is not possible, e.g. when the parameters of the channel are unknown.

Finally, we discussed the deconvolution of noise that does not belong to the class of multiqubit Pauli channels. Our procedure successfully applies also to these cases, while less efficiently and requiring the complete computation (and inversion) of the adjoint PTM.

\section*{Acknowledgments}
This work received support from MIUR Dipartimenti di Eccellenza 2018-2022 Project No. F11I18000680001, from EU H2020 QuantERA ERA-NET Cofund in Quantum Technologies, Quantum Information and Communication with High-dimensional Encoding (QuICHE) under Grant Agreements No. 731473 and No. 101017733, from the U.S. Department of Energy, Office of Science, National Quantum Information Science Research Centers, Superconducting Quantum Materials and Systems Center (SQMS) under Contract No. DE-AC02-07CH11359, and from the National Research Centre for HPC, Big Data and Quantum Computing (ICSC). We thank S. Mangini for helpful discussions.

\appendix*
\section{POSITIVITY OF THE CHARACTERIZATION STATE\label{sec:App}}
In this appendix we show that the operator used in the noise characterization procedure fulfills all the requirements of a density operator, i.e. Hermiticity, unit trace and positive semidefiniteness \citep{book:Nielsen}, namely, that it represents a state of the $n$-qubit system.

Consider the operators of \cref{eq:ChInputState}, with $k = 1,2,...,d^2-1$. While Hermiticity is straightforward, unit trace follows from $\Tr[\mathcal{P}_k] = 0$ for $k \neq 0$.

To prove positive semidefiniteness, we refer to the characterization of density operators in terms of the coherence vector representation discussed in \citep{art:Byrd}. Let $S_m$ be the coefficients of the characteristic polynomial of the $d \times d$ matrix representation of $\rho$,\footnote{In this appendix we use the standard operator framework rather than the vectorized representation.} given by
\begin{equation}
	S_m = \frac{1}{m}\sum_{j=1}^m (-1)^{j-1}\Tr[\rho^j]S_{m-j} \ ,
	\label{eq:RecursiveRelations}
\end{equation}
with $S_0=1$ and $m = 0,1,...,d$.\footnote{Here \cref{eq:RecursiveRelations} is obtained by writing Eq. (24) from \citep{art:Byrd} as a series, with $S_0 = 1$ the coefficient of the highest-order term in the characteristic polynomial of $A$.}

A necessary and sufficient condition for $\rho$ being positive semidefinite is that $S_m \geq 0$ $\forall m$.
\begin{theorem}
The operator $\rho = (1 + \mathcal{P}_k)/d$ is positive semidefinite.
\end{theorem}
\begin{proof}
Consider the operator $A = (1 + \mathcal{P}_k)/2$, which is positive semidefinite if and only if $\rho$ does. A direct computation yields $A^j = A$ and $\Tr[A^j] = d/2$. For $A$, \cref{eq:RecursiveRelations} reads
\begin{equation}
	S_m = \frac{d}{2m}\sum_{j=1}^m (-1)^{j-1}S_{m-j} \ .
	\label{eq:RecursiveRelationsforA}
\end{equation} 
Extracting the first term of the series and collecting a minus sign, we get
\begin{equation}
	S_m = \frac{d}{2m} \left( S_{m-1} - \sum_{j=2}^m (-1)^{j-2}S_{m-j} \right) \ .
\end{equation} 
Translating $j \to j + 1$ (with the sum now running from $1$ to $m-1$) and writing \cref{eq:RecursiveRelationsforA} for $S_{m-1}$, we obtain a recursive relation
\begin{equation}
	S_m = \frac{d}{2m}\left(1 - \frac{2(m-1)}{d} \right)S_{m-1} \ .
\end{equation}
With $\delta = 1 + d/2$, this yields 
\begin{align}
	&S_m > 0 \quad \text{for} \ 0 \leq m < \delta \ , \\
	&S_m = 0 \quad \text{for} \ \delta \leq m \leq d \ ,
\end{align}
which implies that $S_m \geq 0$ $\forall m$. This result does not depend on $k$, nor do the coefficients $S_m$ and the roots of the characteristic polynomial, i.e. the eigenvalue of $A$. Hence, $A$ is positive semidefinite for any choice of $\mathcal{P}_k$ and consequently so is $\rho$.
\end{proof}

\bibliography{refs.bib}

\begin{thebibliography}{30}%
\makeatletter
\providecommand \@ifxundefined [1]{%
 \@ifx{#1\undefined}
}%
\providecommand \@ifnum [1]{%
 \ifnum #1\expandafter \@firstoftwo
 \else \expandafter \@secondoftwo
 \fi
}%
\providecommand \@ifx [1]{%
 \ifx #1\expandafter \@firstoftwo
 \else \expandafter \@secondoftwo
 \fi
}%
\providecommand \natexlab [1]{#1}%
\providecommand \enquote  [1]{``#1''}%
\providecommand \bibnamefont  [1]{#1}%
\providecommand \bibfnamefont [1]{#1}%
\providecommand \citenamefont [1]{#1}%
\providecommand \href@noop [0]{\@secondoftwo}%
\providecommand \href [0]{\begingroup \@sanitize@url \@href}%
\providecommand \@href[1]{\@@startlink{#1}\@@href}%
\providecommand \@@href[1]{\endgroup#1\@@endlink}%
\providecommand \@sanitize@url [0]{\catcode `\\12\catcode `\$12\catcode
  `\&12\catcode `\#12\catcode `\^12\catcode `\_12\catcode `\%12\relax}%
\providecommand \@@startlink[1]{}%
\providecommand \@@endlink[0]{}%
\providecommand \url  [0]{\begingroup\@sanitize@url \@url }%
\providecommand \@url [1]{\endgroup\@href {#1}{\urlprefix }}%
\providecommand \urlprefix  [0]{URL }%
\providecommand \Eprint [0]{\href }%
\providecommand \doibase [0]{https://doi.org/}%
\providecommand \selectlanguage [0]{\@gobble}%
\providecommand \bibinfo  [0]{\@secondoftwo}%
\providecommand \bibfield  [0]{\@secondoftwo}%
\providecommand \translation [1]{[#1]}%
\providecommand \BibitemOpen [0]{}%
\providecommand \bibitemStop [0]{}%
\providecommand \bibitemNoStop [0]{.\EOS\space}%
\providecommand \EOS [0]{\spacefactor3000\relax}%
\providecommand \BibitemShut  [1]{\csname bibitem#1\endcsname}%
\let\auto@bib@innerbib\@empty
\bibitem [{\citenamefont {Nielsen}\ and\ \citenamefont
  {Chuang}(2010)}]{book:Nielsen}%
  \BibitemOpen
  \bibfield  {author} {\bibinfo {author} {\bibfnamefont {M.~A.}\ \bibnamefont
  {Nielsen}}\ and\ \bibinfo {author} {\bibfnamefont {I.~L.}\ \bibnamefont
  {Chuang}},\ }\href {https://doi.org/10.1017/CBO9780511976667} {\emph
  {\bibinfo {title} {Quantum Computation and Quantum Information}}},\ \bibinfo
  {edition} {10th}\ ed.\ (\bibinfo  {publisher} {Cambridge University Press},\
  \bibinfo {year} {2010})\BibitemShut {NoStop}%
\bibitem [{\citenamefont {Ariano}\ \emph {et~al.}(2003)\citenamefont {Ariano},
  \citenamefont {Maccone},\ and\ \citenamefont {Paini}}]{art:QUIT_Tomography}%
  \BibitemOpen
  \bibfield  {author} {\bibinfo {author} {\bibfnamefont {G.~M.~D.}\
  \bibnamefont {Ariano}}, \bibinfo {author} {\bibfnamefont {L.}~\bibnamefont
  {Maccone}},\ and\ \bibinfo {author} {\bibfnamefont {M.}~\bibnamefont
  {Paini}},\ }\bibfield  {title} {\bibinfo {title} {Spin tomography},\ }\href
  {https://doi.org/10.1088/1464-4266/5/1/311} {\bibfield  {journal} {\bibinfo
  {journal} {J. Opt. B}\ }\textbf {\bibinfo {volume} {5}},\ \bibinfo {pages}
  {77} (\bibinfo {year} {2003})}\BibitemShut {NoStop}%
\bibitem [{\citenamefont {Mohseni}\ \emph {et~al.}(2008)\citenamefont
  {Mohseni}, \citenamefont {Rezakhani},\ and\ \citenamefont
  {Lidar}}]{art:Mohseni}%
  \BibitemOpen
  \bibfield  {author} {\bibinfo {author} {\bibfnamefont {M.}~\bibnamefont
  {Mohseni}}, \bibinfo {author} {\bibfnamefont {A.~T.}\ \bibnamefont
  {Rezakhani}},\ and\ \bibinfo {author} {\bibfnamefont {D.~A.}\ \bibnamefont
  {Lidar}},\ }\bibfield  {title} {\bibinfo {title} {Quantum-process tomography:
  Resource analysis of different strategies},\ }\href
  {https://doi.org/10.1103/physreva.77.032322} {\bibfield  {journal} {\bibinfo
  {journal} {Phys. Rev. A}\ }\textbf {\bibinfo {volume} {77}},\ \bibinfo
  {pages} {032322} (\bibinfo {year} {2008})}\BibitemShut {NoStop}%
\bibitem [{\citenamefont {Bongioanni}\ \emph {et~al.}(2010)\citenamefont
  {Bongioanni}, \citenamefont {Sansoni}, \citenamefont {Sciarrino},
  \citenamefont {Vallone},\ and\ \citenamefont {Mataloni}}]{art:Mataloni}%
  \BibitemOpen
  \bibfield  {author} {\bibinfo {author} {\bibfnamefont {I.}~\bibnamefont
  {Bongioanni}}, \bibinfo {author} {\bibfnamefont {L.}~\bibnamefont {Sansoni}},
  \bibinfo {author} {\bibfnamefont {F.}~\bibnamefont {Sciarrino}}, \bibinfo
  {author} {\bibfnamefont {G.}~\bibnamefont {Vallone}},\ and\ \bibinfo {author}
  {\bibfnamefont {P.}~\bibnamefont {Mataloni}},\ }\bibfield  {title} {\bibinfo
  {title} {Experimental quantum process tomography of non-trace-preserving
  maps},\ }\href {https://doi.org/10.1103/PhysRevA.82.042307} {\bibfield
  {journal} {\bibinfo  {journal} {Phys. Rev. A}\ }\textbf {\bibinfo {volume}
  {82}},\ \bibinfo {pages} {042307} (\bibinfo {year} {2010})}\BibitemShut
  {NoStop}%
\bibitem [{\citenamefont {Georgescu}\ \emph {et~al.}(2014)\citenamefont
  {Georgescu}, \citenamefont {Ashhab},\ and\ \citenamefont
  {Nori}}]{art:Georgescu_Simulation}%
  \BibitemOpen
  \bibfield  {author} {\bibinfo {author} {\bibfnamefont {I.~M.}\ \bibnamefont
  {Georgescu}}, \bibinfo {author} {\bibfnamefont {S.}~\bibnamefont {Ashhab}},\
  and\ \bibinfo {author} {\bibfnamefont {F.}~\bibnamefont {Nori}},\ }\bibfield
  {title} {\bibinfo {title} {Quantum simulation},\ }\href
  {https://doi.org/10.1103/RevModPhys.86.153} {\bibfield  {journal} {\bibinfo
  {journal} {Rev. Mod. Phys.}\ }\textbf {\bibinfo {volume} {86}},\ \bibinfo
  {pages} {153} (\bibinfo {year} {2014})}\BibitemShut {NoStop}%
\bibitem [{\citenamefont {Temme}\ \emph {et~al.}(2017)\citenamefont {Temme},
  \citenamefont {Bravyi},\ and\ \citenamefont {Gambetta}}]{art:Temme}%
  \BibitemOpen
  \bibfield  {author} {\bibinfo {author} {\bibfnamefont {K.}~\bibnamefont
  {Temme}}, \bibinfo {author} {\bibfnamefont {S.}~\bibnamefont {Bravyi}},\ and\
  \bibinfo {author} {\bibfnamefont {J.~M.}\ \bibnamefont {Gambetta}},\
  }\bibfield  {title} {\bibinfo {title} {Error mitigation for short-depth
  quantum circuits},\ }\href {https://doi.org/10.1103/PhysRevLett.119.180509}
  {\bibfield  {journal} {\bibinfo  {journal} {Phys. Rev. Lett.}\ }\textbf
  {\bibinfo {volume} {119}},\ \bibinfo {pages} {180509} (\bibinfo {year}
  {2017})}\BibitemShut {NoStop}%
\bibitem [{\citenamefont {Endo}\ \emph {et~al.}(2018)\citenamefont {Endo},
  \citenamefont {Benjamin},\ and\ \citenamefont {Li}}]{art:Endo}%
  \BibitemOpen
  \bibfield  {author} {\bibinfo {author} {\bibfnamefont {S.}~\bibnamefont
  {Endo}}, \bibinfo {author} {\bibfnamefont {S.~C.}\ \bibnamefont {Benjamin}},\
  and\ \bibinfo {author} {\bibfnamefont {Y.}~\bibnamefont {Li}},\ }\bibfield
  {title} {\bibinfo {title} {Practical quantum error mitigation for near-future
  applications},\ }\href {https://doi.org/10.1103/PhysRevX.8.031027} {\bibfield
   {journal} {\bibinfo  {journal} {Phys. Rev. X}\ }\textbf {\bibinfo {volume}
  {8}},\ \bibinfo {pages} {031027} (\bibinfo {year} {2018})}\BibitemShut
  {NoStop}%
\bibitem [{\citenamefont {Endo}\ \emph {et~al.}(2021)\citenamefont {Endo},
  \citenamefont {Cai}, \citenamefont {Benjamin},\ and\ \citenamefont
  {Yuan}}]{art:Endo_Review}%
  \BibitemOpen
  \bibfield  {author} {\bibinfo {author} {\bibfnamefont {S.}~\bibnamefont
  {Endo}}, \bibinfo {author} {\bibfnamefont {Z.}~\bibnamefont {Cai}}, \bibinfo
  {author} {\bibfnamefont {S.~C.}\ \bibnamefont {Benjamin}},\ and\ \bibinfo
  {author} {\bibfnamefont {X.}~\bibnamefont {Yuan}},\ }\bibfield  {title}
  {\bibinfo {title} {Hybrid quantum-classical algorithms and quantum error
  mitigation},\ }\href {https://doi.org/10.7566/jpsj.90.032001} {\bibfield
  {journal} {\bibinfo  {journal} {J. Phys. Soc. Jpn.}\ }\textbf {\bibinfo
  {volume} {90}},\ \bibinfo {pages} {032001} (\bibinfo {year}
  {2021})}\BibitemShut {NoStop}%
\bibitem [{\citenamefont {Mangini}\ \emph {et~al.}(2022)\citenamefont
  {Mangini}, \citenamefont {Maccone},\ and\ \citenamefont
  {Macchiavello}}]{art:QUIT_NoiseDeconvolution}%
  \BibitemOpen
  \bibfield  {author} {\bibinfo {author} {\bibfnamefont {S.}~\bibnamefont
  {Mangini}}, \bibinfo {author} {\bibfnamefont {L.}~\bibnamefont {Maccone}},\
  and\ \bibinfo {author} {\bibfnamefont {C.}~\bibnamefont {Macchiavello}},\
  }\bibfield  {title} {\bibinfo {title} {Qubit noise deconvolution},\ }\href
  {https://doi.org/10.1140/epjqt/s40507-022-00151-0} {\bibfield  {journal}
  {\bibinfo  {journal} {{EPJ} Quantum Technol.}\ }\textbf {\bibinfo {volume}
  {9}},\ \bibinfo {pages} {29} (\bibinfo {year} {2022})}\BibitemShut {NoStop}%
\bibitem [{\citenamefont {Macchiavello}\ and\ \citenamefont
  {Palma}(2002)}]{art:QUIT_CorrelatedNoise}%
  \BibitemOpen
  \bibfield  {author} {\bibinfo {author} {\bibfnamefont {C.}~\bibnamefont
  {Macchiavello}}\ and\ \bibinfo {author} {\bibfnamefont {G.~M.}\ \bibnamefont
  {Palma}},\ }\bibfield  {title} {\bibinfo {title} {Entanglement-enhanced
  information transmission over a quantum channel with correlated noise},\
  }\href {https://doi.org/10.1103/physreva.65.050301} {\bibfield  {journal}
  {\bibinfo  {journal} {Phys. Rev. A}\ }\textbf {\bibinfo {volume} {65}},\
  \bibinfo {pages} {050301(R)} (\bibinfo {year} {2002})}\BibitemShut {NoStop}%
\bibitem [{\citenamefont {Macchiavello}\ \emph {et~al.}(2004)\citenamefont
  {Macchiavello}, \citenamefont {Palma},\ and\ \citenamefont
  {Virmani}}]{art:QUIT_Pauli}%
  \BibitemOpen
  \bibfield  {author} {\bibinfo {author} {\bibfnamefont {C.}~\bibnamefont
  {Macchiavello}}, \bibinfo {author} {\bibfnamefont {G.~M.}\ \bibnamefont
  {Palma}},\ and\ \bibinfo {author} {\bibfnamefont {S.}~\bibnamefont
  {Virmani}},\ }\bibfield  {title} {\bibinfo {title} {Transition behavior in
  the channel capacity of two-quibit channels with memory},\ }\href
  {https://doi.org/10.1103/PhysRevA.69.010303} {\bibfield  {journal} {\bibinfo
  {journal} {Phys. Rev. A}\ }\textbf {\bibinfo {volume} {69}},\ \bibinfo
  {pages} {010303(R)} (\bibinfo {year} {2004})}\BibitemShut {NoStop}%
\bibitem [{\citenamefont {Macchiavello}\ and\ \citenamefont
  {Sacchi}(2016)}]{art:QUIT_QuantumCapacity}%
  \BibitemOpen
  \bibfield  {author} {\bibinfo {author} {\bibfnamefont {C.}~\bibnamefont
  {Macchiavello}}\ and\ \bibinfo {author} {\bibfnamefont {M.~F.}\ \bibnamefont
  {Sacchi}},\ }\bibfield  {title} {\bibinfo {title} {Witnessing quantum
  capacities of correlated channels},\ }\href
  {https://doi.org/10.1103/physreva.94.052333} {\bibfield  {journal} {\bibinfo
  {journal} {Phys. Rev. A}\ }\textbf {\bibinfo {volume} {94}},\ \bibinfo
  {pages} {052333} (\bibinfo {year} {2016})}\BibitemShut {NoStop}%
\bibitem [{\citenamefont {Daems}(2007)}]{art:Daems}%
  \BibitemOpen
  \bibfield  {author} {\bibinfo {author} {\bibfnamefont {D.}~\bibnamefont
  {Daems}},\ }\bibfield  {title} {\bibinfo {title} {Entanglement-enhanced
  transmission of classical information in {P}auli channels with memory: Exact
  solution},\ }\href {https://doi.org/10.1103/PhysRevA.76.012310} {\bibfield
  {journal} {\bibinfo  {journal} {Phys. Rev. A}\ }\textbf {\bibinfo {volume}
  {76}},\ \bibinfo {pages} {012310} (\bibinfo {year} {2007})}\BibitemShut
  {NoStop}%
\bibitem [{\citenamefont {Huang}\ \emph {et~al.}(2011)\citenamefont {Huang},
  \citenamefont {He}, \citenamefont {Lu},\ and\ \citenamefont
  {Zeng}}]{art:Huang}%
  \BibitemOpen
  \bibfield  {author} {\bibinfo {author} {\bibfnamefont {P.}~\bibnamefont
  {Huang}}, \bibinfo {author} {\bibfnamefont {G.}~\bibnamefont {He}}, \bibinfo
  {author} {\bibfnamefont {Y.}~\bibnamefont {Lu}},\ and\ \bibinfo {author}
  {\bibfnamefont {G.}~\bibnamefont {Zeng}},\ }\bibfield  {title} {\bibinfo
  {title} {Quantum capacity of {P}auli channels with memory},\ }\href
  {https://doi.org/10.1088/0031-8949/83/01/015005} {\bibfield  {journal}
  {\bibinfo  {journal} {Phys. Scr.}\ }\textbf {\bibinfo {volume} {83}},\
  \bibinfo {pages} {015005} (\bibinfo {year} {2011})}\BibitemShut {NoStop}%
\bibitem [{\citenamefont {Jiang}\ \emph {et~al.}(2021)\citenamefont {Jiang},
  \citenamefont {Wang},\ and\ \citenamefont {Wang}}]{art:Jiang_InvertibleMaps}%
  \BibitemOpen
  \bibfield  {author} {\bibinfo {author} {\bibfnamefont {J.}~\bibnamefont
  {Jiang}}, \bibinfo {author} {\bibfnamefont {K.}~\bibnamefont {Wang}},\ and\
  \bibinfo {author} {\bibfnamefont {X.}~\bibnamefont {Wang}},\ }\bibfield
  {title} {\bibinfo {title} {Physical implementability of linear maps and its
  application in error mitigation},\ }\href
  {https://doi.org/10.22331/q-2021-12-07-600} {\bibfield  {journal} {\bibinfo
  {journal} {Quantum}\ }\textbf {\bibinfo {volume} {5}},\ \bibinfo {pages}
  {600} (\bibinfo {year} {2021})}\BibitemShut {NoStop}%
\bibitem [{\citenamefont {Lautenbacher}\ \emph {et~al.}(2022)\citenamefont
  {Lautenbacher}, \citenamefont {de~Melo},\ and\ \citenamefont
  {Bernardes}}]{art:Lautenbacher}%
  \BibitemOpen
  \bibfield  {author} {\bibinfo {author} {\bibfnamefont {L.}~\bibnamefont
  {Lautenbacher}}, \bibinfo {author} {\bibfnamefont {F.}~\bibnamefont
  {de~Melo}},\ and\ \bibinfo {author} {\bibfnamefont {N.~K.}\ \bibnamefont
  {Bernardes}},\ }\bibfield  {title} {\bibinfo {title} {Approximating
  invertible maps by recovery channels: Optimality and an application to
  non-{M}arkovian dynamics},\ }\href
  {https://doi.org/10.1103/PhysRevA.105.042421} {\bibfield  {journal} {\bibinfo
   {journal} {Phys. Rev. A}\ }\textbf {\bibinfo {volume} {105}},\ \bibinfo
  {pages} {042421} (\bibinfo {year} {2022})}\BibitemShut {NoStop}%
\bibitem [{\citenamefont {Gily\'en}\ \emph {et~al.}(2022)\citenamefont
  {Gily\'en}, \citenamefont {Lloyd}, \citenamefont {Marvian}, \citenamefont
  {Quek},\ and\ \citenamefont {Wilde}}]{art:Lloyd_PetzRecovery}%
  \BibitemOpen
  \bibfield  {author} {\bibinfo {author} {\bibfnamefont {A.}~\bibnamefont
  {Gily\'en}}, \bibinfo {author} {\bibfnamefont {S.}~\bibnamefont {Lloyd}},
  \bibinfo {author} {\bibfnamefont {I.}~\bibnamefont {Marvian}}, \bibinfo
  {author} {\bibfnamefont {Y.}~\bibnamefont {Quek}},\ and\ \bibinfo {author}
  {\bibfnamefont {M.~M.}\ \bibnamefont {Wilde}},\ }\bibfield  {title} {\bibinfo
  {title} {Quantum algorithm for {P}etz recovery channels and pretty good
  measurements},\ }\href {https://doi.org/10.1103/PhysRevLett.128.220502}
  {\bibfield  {journal} {\bibinfo  {journal} {Phys. Rev. Lett.}\ }\textbf
  {\bibinfo {volume} {128}},\ \bibinfo {pages} {220502} (\bibinfo {year}
  {2022})}\BibitemShut {NoStop}%
\bibitem [{\citenamefont {Flammia}\ and\ \citenamefont
  {Wallman}(2020)}]{art:Flammia}%
  \BibitemOpen
  \bibfield  {author} {\bibinfo {author} {\bibfnamefont {S.~T.}\ \bibnamefont
  {Flammia}}\ and\ \bibinfo {author} {\bibfnamefont {J.~J.}\ \bibnamefont
  {Wallman}},\ }\bibfield  {title} {\bibinfo {title} {Efficient estimation of
  {P}auli channels},\ }\href {https://doi.org/10.1145/3408039} {\bibfield
  {journal} {\bibinfo  {journal} {{ACM} Trans. Quantum Comput.}\ }\textbf
  {\bibinfo {volume} {1}},\ \bibinfo {pages} {1} (\bibinfo {year}
  {2020})}\BibitemShut {NoStop}%
\bibitem [{\citenamefont {Greenbaum}(2015)}]{art:Greenbaum}%
  \BibitemOpen
  \bibfield  {author} {\bibinfo {author} {\bibfnamefont {D.}~\bibnamefont
  {Greenbaum}},\ }\href@noop {} {\bibinfo {title} {Introduction to quantum gate
  set tomography}} (\bibinfo {year} {2015}),\ \Eprint
  {https://arxiv.org/abs/1509.02921} {arXiv:1509.02921 [quant-ph]} \BibitemShut
  {NoStop}%
\bibitem [{\citenamefont {Nielsen}\ \emph {et~al.}(2021)\citenamefont
  {Nielsen}, \citenamefont {Gamble}, \citenamefont {Rudinger}, \citenamefont
  {Scholten}, \citenamefont {Young},\ and\ \citenamefont
  {Blume-Kohout}}]{art:Nielsen_GateTomography}%
  \BibitemOpen
  \bibfield  {author} {\bibinfo {author} {\bibfnamefont {E.}~\bibnamefont
  {Nielsen}}, \bibinfo {author} {\bibfnamefont {J.~K.}\ \bibnamefont {Gamble}},
  \bibinfo {author} {\bibfnamefont {K.}~\bibnamefont {Rudinger}}, \bibinfo
  {author} {\bibfnamefont {T.}~\bibnamefont {Scholten}}, \bibinfo {author}
  {\bibfnamefont {K.}~\bibnamefont {Young}},\ and\ \bibinfo {author}
  {\bibfnamefont {R.}~\bibnamefont {Blume-Kohout}},\ }\bibfield  {title}
  {\bibinfo {title} {Gate set tomography},\ }\href
  {https://doi.org/10.22331/q-2021-10-05-557} {\bibfield  {journal} {\bibinfo
  {journal} {Quantum}\ }\textbf {\bibinfo {volume} {5}},\ \bibinfo {pages}
  {557} (\bibinfo {year} {2021})}\BibitemShut {NoStop}%
\bibitem [{\citenamefont {Wood}\ \emph {et~al.}(2015)\citenamefont {Wood},
  \citenamefont {Biamonte},\ and\ \citenamefont
  {Cory}}]{art:Wood_ChannelRepresentations}%
  \BibitemOpen
  \bibfield  {author} {\bibinfo {author} {\bibfnamefont {C.~J.}\ \bibnamefont
  {Wood}}, \bibinfo {author} {\bibfnamefont {J.~D.}\ \bibnamefont {Biamonte}},\
  and\ \bibinfo {author} {\bibfnamefont {D.~G.}\ \bibnamefont {Cory}},\
  }\bibfield  {title} {\bibinfo {title} {Tensor networks and graphical calculus
  for open quantum systems},\ }\href
  {https://doi.org/https://doi.org/10.26421/QIC15.9-10-3} {\bibfield  {journal}
  {\bibinfo  {journal} {Quantum Inf. Comput.}\ }\textbf {\bibinfo {volume}
  {15}},\ \bibinfo {pages} {759} (\bibinfo {year} {2015})}\BibitemShut
  {NoStop}%
\bibitem [{\citenamefont {B{\'{e}}ny}(2017)}]{art:Beny}%
  \BibitemOpen
  \bibfield  {author} {\bibinfo {author} {\bibfnamefont {C.}~\bibnamefont
  {B{\'{e}}ny}},\ }\bibfield  {title} {\bibinfo {title} {Quantum
  deconvolution},\ }\href {https://doi.org/10.1007/s11128-017-1796-3}
  {\bibfield  {journal} {\bibinfo  {journal} {Quantum Inf. Process.}\ }\textbf
  {\bibinfo {volume} {17}},\ \bibinfo {pages} {26} (\bibinfo {year}
  {2017})}\BibitemShut {NoStop}%
\bibitem [{\citenamefont {Audenaert}\ and\ \citenamefont
  {Scheel}(2008)}]{art:Audenaert_RandomUnitary}%
  \BibitemOpen
  \bibfield  {author} {\bibinfo {author} {\bibfnamefont {K.~M.~R.}\
  \bibnamefont {Audenaert}}\ and\ \bibinfo {author} {\bibfnamefont
  {S.}~\bibnamefont {Scheel}},\ }\bibfield  {title} {\bibinfo {title} {On
  random unitary channels},\ }\href
  {https://doi.org/10.1088/1367-2630/10/2/023011} {\bibfield  {journal}
  {\bibinfo  {journal} {New J. Phys.}\ }\textbf {\bibinfo {volume} {10}},\
  \bibinfo {pages} {023011} (\bibinfo {year} {2008})}\BibitemShut {NoStop}%
\bibitem [{\citenamefont
  {Siudzi{\'{n}}ska}(2020)}]{art:Siudzinska_GeneralizedPauliChannels}%
  \BibitemOpen
  \bibfield  {author} {\bibinfo {author} {\bibfnamefont {K.}~\bibnamefont
  {Siudzi{\'{n}}ska}},\ }\bibfield  {title} {\bibinfo {title} {Classical
  capacity of generalized {P}auli channels},\ }\href
  {https://doi.org/10.1088/1751-8121/abb276} {\bibfield  {journal} {\bibinfo
  {journal} {J. Phys. A Math. Theor.}\ }\textbf {\bibinfo {volume} {53}},\
  \bibinfo {pages} {445301} (\bibinfo {year} {2020})}\BibitemShut {NoStop}%
\bibitem [{\citenamefont {Cai}\ \emph {et~al.}(2020)\citenamefont {Cai},
  \citenamefont {Xu},\ and\ \citenamefont
  {Benjamin}}]{art:Cai_PauliConjugation}%
  \BibitemOpen
  \bibfield  {author} {\bibinfo {author} {\bibfnamefont {Z.}~\bibnamefont
  {Cai}}, \bibinfo {author} {\bibfnamefont {X.}~\bibnamefont {Xu}},\ and\
  \bibinfo {author} {\bibfnamefont {S.~C.}\ \bibnamefont {Benjamin}},\
  }\bibfield  {title} {\bibinfo {title} {Mitigating coherent noise using
  {P}auli conjugation},\ }\href {https://doi.org/10.1038/s41534-019-0233-0}
  {\bibfield  {journal} {\bibinfo  {journal} {NPJ Quantum. Inf.}\ }\textbf
  {\bibinfo {volume} {6}},\ \bibinfo {pages} {17} (\bibinfo {year}
  {2020})}\BibitemShut {NoStop}%
\bibitem [{\citenamefont {Roncallo}\ \emph {et~al.}(2022)\citenamefont
  {Roncallo}, \citenamefont {Maccone},\ and\ \citenamefont
  {Macchiavello}}]{art:QUIT_InPreparation}%
  \BibitemOpen
  \bibfield  {author} {\bibinfo {author} {\bibfnamefont {S.}~\bibnamefont
  {Roncallo}}, \bibinfo {author} {\bibfnamefont {L.}~\bibnamefont {Maccone}},\
  and\ \bibinfo {author} {\bibfnamefont {C.}~\bibnamefont {Macchiavello}},\
  }\href@noop {} {\bibinfo {title} {Pauli transfer matrix direct
  reconstruction}} (\bibinfo {year} {2022}),\ \Eprint
  {https://arxiv.org/abs/2212.11968} {arXiv:2212.11968 [quant-ph]} \BibitemShut
  {NoStop}%
\bibitem [{\citenamefont {D'Arrigo}\ \emph {et~al.}(2013)\citenamefont
  {D'Arrigo}, \citenamefont {Benenti}, \citenamefont {Falci},\ and\
  \citenamefont {Macchiavello}}]{art:QUIT_FullyCorrelatedDamping}%
  \BibitemOpen
  \bibfield  {author} {\bibinfo {author} {\bibfnamefont {A.}~\bibnamefont
  {D'Arrigo}}, \bibinfo {author} {\bibfnamefont {G.}~\bibnamefont {Benenti}},
  \bibinfo {author} {\bibfnamefont {G.}~\bibnamefont {Falci}},\ and\ \bibinfo
  {author} {\bibfnamefont {C.}~\bibnamefont {Macchiavello}},\ }\bibfield
  {title} {\bibinfo {title} {Classical and quantum capacities of a fully
  correlated amplitude damping channel},\ }\href
  {https://doi.org/10.1103/PhysRevA.88.042337} {\bibfield  {journal} {\bibinfo
  {journal} {Phys. Rev. A}\ }\textbf {\bibinfo {volume} {88}},\ \bibinfo
  {pages} {042337} (\bibinfo {year} {2013})}\BibitemShut {NoStop}%
\bibitem [{\citenamefont {D'Arrigo}\ \emph {et~al.}(2015)\citenamefont
  {D'Arrigo}, \citenamefont {Benenti}, \citenamefont {Falci},\ and\
  \citenamefont {Macchiavello}}]{art:QUIT_CorrelatedDamping}%
  \BibitemOpen
  \bibfield  {author} {\bibinfo {author} {\bibfnamefont {A.}~\bibnamefont
  {D'Arrigo}}, \bibinfo {author} {\bibfnamefont {G.}~\bibnamefont {Benenti}},
  \bibinfo {author} {\bibfnamefont {G.}~\bibnamefont {Falci}},\ and\ \bibinfo
  {author} {\bibfnamefont {C.}~\bibnamefont {Macchiavello}},\ }\bibfield
  {title} {\bibinfo {title} {Information transmission over an amplitude damping
  channel with an arbitrary degree of memory},\ }\href
  {https://doi.org/10.1103/PhysRevA.92.062342} {\bibfield  {journal} {\bibinfo
  {journal} {Phys. Rev. A}\ }\textbf {\bibinfo {volume} {92}},\ \bibinfo
  {pages} {062342} (\bibinfo {year} {2015})}\BibitemShut {NoStop}%
\bibitem [{\citenamefont {Hamada}(2002)}]{art:Hamada}%
  \BibitemOpen
  \bibfield  {author} {\bibinfo {author} {\bibfnamefont {M.}~\bibnamefont
  {Hamada}},\ }\bibfield  {title} {\bibinfo {title} {A lower bound on the
  quantum capacity of channels with correlated errors},\ }\href
  {https://doi.org/10.1063/1.1495537} {\bibfield  {journal} {\bibinfo
  {journal} {J. Math. Phys.}\ }\textbf {\bibinfo {volume} {43}},\ \bibinfo
  {pages} {4382} (\bibinfo {year} {2002})}\BibitemShut {NoStop}%
\bibitem [{\citenamefont {Byrd}\ and\ \citenamefont
  {Khaneja}(2003)}]{art:Byrd}%
  \BibitemOpen
  \bibfield  {author} {\bibinfo {author} {\bibfnamefont {M.~S.}\ \bibnamefont
  {Byrd}}\ and\ \bibinfo {author} {\bibfnamefont {N.}~\bibnamefont {Khaneja}},\
  }\bibfield  {title} {\bibinfo {title} {Characterization of the positivity of
  the density matrix in terms of the coherence vector representation},\ }\href
  {https://doi.org/10.1103/PhysRevA.68.062322} {\bibfield  {journal} {\bibinfo
  {journal} {Phys. Rev. A}\ }\textbf {\bibinfo {volume} {68}},\ \bibinfo
  {pages} {062322} (\bibinfo {year} {2003})}\BibitemShut {NoStop}%
\end{thebibliography}%
\end{document}